\documentclass[a4paper,UKenglish,cleveref,autoref,thm-restate]{lipics-v2021}

\title{An $11/6$-Approximation Algorithm for Vertex Cover on String Graphs}
\titlerunning{An $11/6$-Approximation Algorithm for Vertex Cover on String Graphs}

\author{\'{E}douard Bonnet}{Univ Lyon, CNRS, ENS de Lyon, Université Claude Bernard Lyon 1, LIP UMR5668, France \and \url{http://perso.ens-lyon.fr/edouard.bonnet/}}{edouard.bonnet@ens-lyon.fr}{https://orcid.org/0000-0002-1653-5822}{}
\author{Paweł Rzążewski}{Warsaw University of Technology \& University of Warsaw, Warsaw, Poland}{pawel.rzazewski@pw.edu.pl}{https://orcid.org/0000-0001-7696-3848}{}

\authorrunning{\'E. Bonnet and P. Rz\k{a}\.{z}ewski}

\Copyright{Édouard Bonnet and Paweł Rz\k{a}\.{z}ewski}

\category{}

\relatedversion{}

\supplement{}

\acknowledgements{}

\nolinenumbers 

\hideLIPIcs  

\EventEditors{John Q. Open and Joan R. Access}
\EventNoEds{2}
\EventLongTitle{42nd Conference on Very Important Topics (CVIT 2016)}
\EventShortTitle{CVIT 2016}
\EventAcronym{CVIT}
\EventYear{2016}
\EventDate{December 24--27, 2016}
\EventLocation{Little Whinging, United Kingdom}
\EventLogo{}
\SeriesVolume{42}
\ArticleNo{23}

\usepackage[utf8]{inputenc}  

\usepackage[T1]{fontenc}
\usepackage{lmodern}

\usepackage[colorinlistoftodos,bordercolor=orange,backgroundcolor=orange!20,linecolor=orange,textsize=normalsize]{todonotes}

\usepackage{amsmath}  
\usepackage{bbm,bm}
\usepackage{accents}
\usepackage{complexity}
\usepackage[numbers,sort&compress]{natbib}
\usepackage{booktabs}
\usepackage{paralist}

\usepackage{pgfplots}
\usepackage{xspace}
\usepackage{tikz}
\usepackage{tikz-3dplot}

\usepackage[ruled,vlined,linesnumbered]{algorithm2e}

\usetikzlibrary{fit}
\usetikzlibrary{arrows}
\usetikzlibrary{patterns}
\usetikzlibrary{calc}
\usetikzlibrary{shapes}
\usetikzlibrary{positioning}
\usetikzlibrary{math}
\usetikzlibrary{shapes.symbols, shapes.geometric}
\usetikzlibrary{decorations.pathreplacing,calligraphy}
\usetikzlibrary{decorations.pathmorphing, backgrounds}
\usepackage[scr=boondox,scrscaled=1.05]{mathalfa}

\newcommand{\vc}{\textsc{Vertex Cover}\xspace}
\newcommand{\vcp}{\mathsf{vc}\xspace}

\renewcommand{\geq}{\geqslant}
\renewcommand{\leq}{\leqslant}

\newtheorem{question}{Question}

\newcommand{\N}{{\mathbb N}}

\begin{document}

\maketitle

\begin{abstract}
  We present a~1.8334-approximation algorithm for \textsc{Vertex Cover} on string graphs given with a~representation, which takes polynomial time in the size of the representation; the exact approximation factor is $\frac{11}{6}$.
  Recently, the barrier of~2 was broken by Lokshtanov et al. [SoGC '24] with a~1.9999-approximation algorithm.
  Thus we increase by three orders of magnitude the distance of the approximation ratio to the trivial bound of~2.
  Our algorithm is very simple.
  The intricacies reside in its analysis, where we mainly establish that string graphs without odd cycles of length at~most~11 are 8-colorable.
  Previously, Chudnovsky, Scott, and Seymour [JCTB '21] showed that string graphs without odd cycles of length at~most~7 are 80-colorable, and string graphs without odd cycles of length at~most~5 have bounded chromatic number.
\end{abstract}

\section{Introduction}\label{sec:intro}

The \textsc{Vertex Cover} problem\footnote{As an optimization problem, given a~graph $G$, find a~smallest possible vertex subset $S$ such that $G-S$ (i.e., the subgraph of $G$ induced by all the vertices \emph{not} in $S$) is edgeless.} is one of Karp's 21 \textsf{NP}-complete problems~\cite{Karp72}.
While it admits several easy polynomial-time 2-approximation algorithms, it is an open question if an approximation factor of $2-\varepsilon$ can be achieved for some $\varepsilon > 0$.
If the unique games conjecture (UGC) holds, then the answer to the latter question is negative~\cite{DBLP:journals/jcss/KhotR08}.
Under the sole $\mathsf{P} \neq \mathsf{NP}$ assumption, it is currently only known that \textsc{Vertex Cover} cannot be approximated within ratio less than~$\sqrt 2$~\cite{KhotMS17,DinurKKMS18,KhotMS18}.

On several graphs classes, better approximation algorithms of \textsc{Vertex Cover} exist.
For instance, this problem can be exactly solved in polynomial-time on bipartite graphs, as it reduces to a~maximum flow problem~\cite{Konig}.
It also admits a~polynomial-time approximation scheme (PTAS) on planar graphs by Baker's technique~\cite{Baker94}.
On the contrary, a~$(2-\varepsilon)$-approximation algorithm for \textsc{Vertex Cover} with $\varepsilon \in (0,1/2]$ on triangle-free graphs would imply the same holding for general graphs.
Indeed, one can observe that any vertex cover intersects each triangle on at~least two vertices.
One can thus repeatedly include all three vertices of a~triangle in the approximate solution.
Once the input graph has no triangle left, one calls the $(2-\varepsilon)$-approximation algorithm.
This strategy can be seen to achieve approximation factor $\max(3/2,2-\varepsilon)$ on general graphs.
In particular, the former algorithm on triangle-free graphs would refute the UGC.
Thus we say that $(2-\varepsilon)$-approximating \textsc{Vertex Cover} on triangle-free graphs is UGC-hard.
It would be interesting to establish a~dichotomy splitting the hereditary\footnote{closed under vertex removals} graph classes on which we know a~$(2-\varepsilon)$-approximation algorithm from those on which this task is UGC-hard.

The class of~\emph{string graphs} consists of every intersection graph of connected sets in a~planar graph.
Alternatively string graphs are the intersection graphs of (non-self-intersecting, non-closed) curves in the planes, called \emph{strings}.
The set of strings is called a~\emph{geometric representation} or \emph{string representation}.
The geometric representation may be convenient for drawing and in the proofs, but it is not so easy to handle as input.
For that, we will also adopt a~more combinatorial approach, based on the first definition of string graphs. 
A~\emph{representation} of a~string graph $G$ is a~planar graph $P$ and as many connected vertex subsets of $P$ as $G$ has vertices;
two vertices are joined by an edge if and only if their corresponding subsets intersect.
A caveat here is that finding such a representation is \textsf{NP}-hard~\cite{DBLP:journals/jct/Kratochvil91a}.
Furthermore, some string graphs $G$ require representations $P$ whose number of vertices is exponential in $|V(G)|$~\cite{Kratochvil91}. 

Recently Lokshtanov, Panolan, Saurabh, Xue, and Zehavi~\cite{Lokshtanov24} presented a~$(2-\varepsilon)$-approximation algorithm for \textsc{Vertex Cover} on string graphs, for some small positive $\varepsilon \approx 10^{-4}$.
In this paper, we give a~simpler algorithm with improved approximation factor $11/6=1.8333\dots$

\begin{restatable}{theorem}{thmvcstrings}\label{thm:vcstrings}
\vc admits an~$\frac{11}{6}$-approximation algorithm in string graphs given with a~representation, whose running time is polynomial in the size of a representation.
\end{restatable}

Our algorithm works as follows.
We first get rid of all the odd cycles of length at~most~11.
This is done as in the abovementioned reduction to triangle-free graphs.
While there is an odd cycle $C$ with $|V(C)| \leqslant 11$, any vertex cover contains at~least~6 vertices of~$C$.
We thus include all vertices of $C$ in our approximate solution, and remove them from the graph.
This ensures a~$11/6$ approximation factor, if such a~factor can be achieved in the resulting graph.
We can thus assume that our input graph has \emph{odd girth}\footnote{The (odd) girth of a~graph $G$ is the length of a~shortest (odd) cycle of~$G$.} more than 11. 
As the standard linear-programming (LP) formulation of~\textsc{Vertex Cover} is half-integral \cite{DBLP:journals/mp/NemhauserT74}, we can further assume that the vertex cover contains at~least half of the vertices.
Note that these opening steps are also present in the paper of Lokshtanov et al.~\cite{Lokshtanov24}.

We now deviate from the previous algorithm, and simply bound the chromatic number of string graphs of odd~girth larger than~11. 
\begin{theorem}\label{thm:coloring}
Every string graph of odd girth larger than 11 is 8-colorable, and given a~representation, an 8-coloring can be found in time polynomial in the representation.
\end{theorem}
A~largest color class in a~proper 8-coloring has size at~least $n/8$ in an $n$-vertex graph.
Thus, its complement is a~vertex cover of size at~most $7n/8$.
As the optimum solution has size at~least $n/2$, this yields a~$7/4=1.75$ factor, and we conclude.
Therefore, the main technical content of the paper is the proof of~\cref{thm:coloring}.

\subparagraph{Outline of the proof of~\cref{thm:coloring}.}
Our strategy goes as follows.
We make a~breadth-first search (BFS) from some arbitrary vertex $u_0$ in each connected component of the graph~$G$.
We reserve colors $1, 2, 3, 4$ for the even-indexed BFS layer, and $5, 6, 7, 8$ for the odd-indexed BFS layer.
We thus need to 4-color each connected component $X$ of each subgraph induced by a~single layer.
Let $R$ be a~string representation of the entire graph, and $R[X]$ its restriction to $X$.
As any face in the arrangement $R[X]$ can be made the infinite face, we can assume that the string of $u_0$ lies on the infinite face of~$R[X]$. 
Now $R[X]$ has a~nice property: The minimal topological disk $D$ that contains it is such that each string of $R[X]$ intersects at~least~one string (in $R \setminus R[X]$) that itself crosses $\partial D$, the boundary of $D$.
Indeed, such a~string is given by a~neighbor in the previous layer.

We now consider the string representation $R_H$ made by $X$ and at~least~one neighbor of the previous layer for each vertex of $X$, and intersect it by $D$.
Each string of $X$ remains in one piece, whereas the strings of the previous layer are possibly split into several strings in $D$.
Let $H$ be the intersection graph of~$R_H$.
We further layer $X$ with the distance in $H$ to a~fixed vertex $w \in V(H) \setminus V(X)$.
Let $X_k$ be the vertices of~$H$ at distance exactly~$k$ from $w$ in~$H$.

We are left with proving that $H[X_k]$ is bipartite.
For the sake of contradiction, we assume that there is an odd cycle $C$ in $H[X_k]$.
Our goal is to show that this odd cycle is contained in a~ball of small radius around some vertex, yielding a~contradiction in the form of a~short odd cycle.
For that, we establish that there is another odd cycle $C'$ (built from~$C$) such that the string of~$w$ is contained in a face $F$ made by few strings of~$N_H[C']$, with most of the rest of $C'$ not intersecting~$F$.
A~string defining~$F$ is then at a~small distance of \emph{every} string of~$C'$ since a~(shortest) path from $w$ to a~vertex of the rest of~$C'$ has to cross the boundary of $F$.
After which, the path has only a~constant number of steps to reach its target since vertices of $N_H[C']$ are at distance~$k-1$, $k$, or $k+1$ from $w$. 
The short odd cycle in $H$ implies the existence of a~short odd cycle in~$G$, a~contradiction. 

\subparagraph{Chromatic number of intersection graphs without short (odd) cycles.}
The chromatic number of intersection graphs of a~given girth has a~long history.
Erd\H{o}s and Gy\'arf\'as asked if girth-4 (i.e., triangle-free) segment intersection graphs have bounded chromatic number~\cite{Gyarfas87}.
Kostochka and Nešetřil~\cite{Kostochka95} raised the same question for 1-string graphs.\footnote{The 1-strings are strings every pair of which intersects at~most~once; note that it is not a property of particular objects but rather their arrangement.}
Both questions were answered in the negative by Pawlik et al.~\cite{Pawlik14}.
It was further shown by Walczak~\cite{Walczak15} that there are triangle-free segment intersection graphs without independent sets of size linear in the number of vertices.
Nevertheless, Kostochka and Nešetřil~\cite{Kostochka98} proved that 1-string graphs of girth at~least~8 are 3-colorable, and asked~\cite[Problem 3]{Kostochka98} whether~8 could be replaced by~5.
This has been confirmed for \emph{outer} 1-strings by Das, Mukherjee, and Sahoo~\cite{Das21}.

More relevant to our \textsc{Vertex Cover} application, the chromatic number of (1-)string graphs of a~given odd girth has also been considered.
(Note indeed that, in the design of a~$(2-\varepsilon)$-approximation algorithm for \textsc{Vertex Cover}, whereas one can remove odd cycles up to any fixed size, the same trick does not apply to 4-vertex cycles.)
McGuinness established that 1-string graphs of odd girth at~least~7 have bounded chromatic number~\cite{McGuinness00,McGuinness01}.
Chudnovsky, Scott, and Seymour~\cite{Chudnovsky21} further showed that string graphs of odd girth at~least~9 are 80-colorable, and string graphs of odd girth at~least~7 have bounded chromatic number.
With proper adjustments, the former result can be turned into an algorithm that inputs a~representation of any string graph $G$ of odd girth at~least~9, and outputs an 80-coloring of $G$ in time polynomial in the representation.
However, this would yield a~significantly worse approximation factor for \textsc{Vertex Cover}.

\subparagraph{Limits of the method.}
The following observation summarizes the trade-off between required odd girth and effective upper bound on the chromatic number, in how they impact the approximation ratio. 

\begin{observation}\label{obs:trade-off}
  Let $\mathcal C$ be a~class such that there is a~polynomial-time algorithm to \mbox{$c$-color} the graphs of $\mathcal C$ of odd~girth at~least an odd positive integer~$g$.
  Then \textsc{Vertex Cover} admits a~polynomial-time $2 \max(\frac{g-2}{g-1}, \frac{c-1}{c})$-approximation in $\mathcal C$. 
\end{observation}

We apply~\cref{obs:trade-off} on represented string graphs with $g=13$ and $c=8$.
Showing a~counterpart of~\cref{thm:coloring} with $g=11$ and $c \leqslant 10$ would improve the ratio to $9/5=1.8$, and with $g=9$ and $c \leqslant 8$, to $7/4=1.75$.
Past this point, our 8-coloring starts being the bottleneck.
The last possible step of this method would be to get parameters $g=7$ and $c \leqslant 6$, for a~ratio of~$5/3=1.666\ldots$
Indeed we recall that string graphs of odd girth at~least~5 have unbounded chromatic number~\cite{Pawlik14}, and do not necessarily have linear-size independent sets~\cite{Walczak15}.
On the complexity side, we do not expect that \textsc{Vertex Cover} is \textsf{APX}-hard (i.e., \textsf{NP}-hard to approximate within some constant ratio $r > 1$) on represented string graphs as a~quasipolynomial-time approximation scheme (QPTAS) exists~\cite{Adamaszek19,HarPeled23}.
However, if a~PTAS also exists, it will have to be found with a~different approach than ours.

Let us emphasize that \cref{thm:coloring} is the only place where our algorithm actually requires a representation of the input to be given.
Thus, an algorithm coloring string graphs with no short odd cycles, that works directly on the input graph~$G$ (not its representation), would yield an approximation algorithm for \textsc{Vertex Cover} in string graphs whose complexity is polynomial in the number of vertices of~$G$.
However, if we are only interested in the decision variant of the problem, the existential statement in \cref{thm:coloring} is sufficient to get the following result.

\begin{theorem}\label{thm:promise}
Given a graph $G$ and an integer $k$, in time polynomial in $|V(G)|$ we can distinguish the following cases:
\begin{compactenum}
\item $G$ is a string graph and $\vc(G) \leq k$, and
\item $G$ is a string graph and $\vc(G) > 11k/6$.
\end{compactenum}
If none of the cases applies, the algorithm terminates but its output is unspecified.
\end{theorem}
\cref{thm:promise} can be seen as a polynomial-time solution to the \emph{promise variant} of \textsc{Vertex Cover}, where the input is guaranteed to satisfy one of the above conditions; see~\cite{DBLP:journals/jacm/RomeroWZ23,DBLP:journals/siglog/KrokhinO22}.

Let us remark that Lokshtanov et al.~\cite{Lokshtanov24} claim that their algorithm can actually return a~1.9999-approximate solution in time polynomial in the number of vertices of the input graph but this is, unfortunately, not true.
Indeed, one of the crucial steps in their argument is the application of a result of Lee~\cite[Theorem 4.2]{DBLP:conf/innovations/Lee17}, whose proof has recently been realized to be flawed\footnote{The inequality at the sixth line of the proof of Lemma 4.14 need not hold.}~\cite{Leeprivate,DBLP:journals/corr/abs-2312-07962}.
Seemingly the result \cite[Theorem 4.2]{DBLP:conf/innovations/Lee17} can be reproven in a different way but only for \emph{region intersection graphs} (a~generalization of string graphs)~\cite{Daviesprivate}, and the new approach requires that the representation is given (and still the final approximation ratio is much closer to~2 than the one given by~\cref{thm:vcstrings}).
Similarly, Lokshtanov et al.~\cite{Lokshtanov24} claim that their approach applies for all proper \emph{induced-minor-closed} classes, but the problem lies in the same place: we are not aware of any way to fix the proof of Lee~\cite[Theorem 4.2]{DBLP:conf/innovations/Lee17} in such a general setting.

We leave as open questions if \textsc{Vertex Cover} admits, for some constant $\varepsilon > 0$, a~$(2-\varepsilon)$-approximation algorithm on string graphs given without representation (which is very likely to be the case), and on classes excluding a~fixed graph~$H$ as an induced minor.\footnote{i.e, $H$ cannot be obtained from graphs of the class by removing vertices and contracting edges} 

\subparagraph{Further discussion and questions.}
An interesting notion to attack our question of which classes admit a~less-than-2 approximation factor is that of $r$-controlledness (see~\cite{Chudnovsky21}).
For an integer $r\geq 1$, a~class $\mathcal C$ is \emph{$r$-controlled} if there exists a function $f : 
\N \to \N$, such that for every $G \in \mathcal{C}$, the chromatic number of $G$ is bounded by $f(\chi^r(G))$,
where $\chi^r(G)$ is the maximum chromatic number of a ball $B^r(v)$ of~radius~$r$ centered at a vertex $v$ of $G$ (i.e., the graph induced by the vertices at distance at~most~$r$ from~$v$).
The function $f$ is in the definition above is called a \emph{binding function}.

In particular, every \emph{$\chi$-bounded} class $\mathcal C$, i.e., a class whose chromatic number $\chi$ is upper bounded by a~function of its clique number $\omega$, is 1-controlled.
Indeed, if $f$ is a~$\chi$-binding function for $\mathcal C$, we have $\chi^1(G) \geqslant \omega(G)$ and $\chi(G) \leqslant f(\omega(G))$ for every $G \in \mathcal C$.
Thus $\chi(G) \leqslant f(\chi^1(G))$ (note that $f$ can be chosen non-decreasing).
We further say that $\mathcal C$ is \emph{effectively $r$-controlled} if there is a~polynomial-time algorithm that turns any $\chi^r(G)$-coloring of $B^r(v)$ for any $G \in \mathcal C$ and each $v \in V(G)$, into an $f(\chi^r(G))$-coloring of~$G$. 
Note that $\mathcal C$ is $r$-controlled and $G \in \mathcal{C}$ has odd girth larger that $2r+1$, then for every $v \in V(G)$, the subgraph induced by $B^r(v)$ is bipartite.
Consequently, similarly to~\cref{obs:trade-off}, we obtain the following.

\begin{observation}\label{obs:controlled}
   \textsc{Vertex Cover} admits a~polynomial-time $\max(\frac{2r+1}{r+1},2(1-\frac{1}{f(2)}))$-approxi\-mation algorithm on any effectively $r$-controlled class with binding function~$f$. 
\end{observation}

Being $r$-controlled for at~least~some~$r$ is quite an inclusive property among ``structured'' graph classes~\cite{Chudnovsky21}.
Thus we wonder the following.

\begin{question}\label{q:controlled}
  Which hereditary classes $\mathcal C$ are not (effectively) $r$-controlled for any~$r$, and yet \textsc{Vertex Cover} admits a~polynomial-time $(2-\varepsilon)$-approximation algorithm in~$\mathcal C$, for some constant $\varepsilon > 0$?   
\end{question}
A~candidate class for~\cref{q:controlled} is made by all graphs of girth at~least~5.
There is no~$r$ for which this class is $r$-controlled, and to our knowledge, no $(2-\varepsilon)$-approximation algorithm nor UGC-hardness proof.

As we already mentioned, intersection graphs of 1-strings with girth at~least~8 are \mbox{3-chromatic}~\cite{Kostochka98}.
Indeed, these graphs are 2-degenerate, i.e., (all their induced subgraphs) have a~vertex of degree at~most~2, which gives a~recursive 3-coloring strategy.
While string graphs of large odd girth are not $d$-degenerate for any $d$ (they contain arbitrarily large bipartite complete graphs), we conjecture that they are nevertheless 3-colorable.

\begin{conjecture}\label{conj:3col}
  There is an integer $k$ such that the class of string graphs of odd girth at~least~$k$ is 3-chromatic.
\end{conjecture}
A~stronger form of~\cref{conj:3col} is that it holds for $k=7$.

\section{Definitions and preparatory lemmas}

If $i \leqslant j$ are two non-negative integers, we denote by $[i,j]$ the set $\{i,i+1,\ldots,j-1,j\}$, and $[i]$ is a~short-hand for $[1,i]$.

\subparagraph{Graphs and odd cycles}

We denote by $V(G)$ and $E(G)$ the set of vertices and edges of a graph $G$, respectively.
A~graph $H$ is an~\emph{induced subgraph} (resp.~\emph{subgraph}) of a~graph $G$ if $H$ can be obtained from $G$ by vertex deletions (resp.~by vertex and edge deletions).
For $S \subseteq V(G)$, the \emph{subgraph of $G$ induced by $S$}, denoted $G[S]$, is obtained by removing from $G$ all the vertices that are not in $S$.
Then $G-S$ is a short-hand for $G[V(G)\setminus S]$.
A~set $X \subseteq V(G)$ is connected (in $G$) if $G[X]$ has a~single connected component. 

If $C$ is a~cycle, once a~direction along the cycle is fixed, we may denote by $C[x \rightarrow y]$ the subpath of $C$ from $x \in V(C)$ to $y \in V(C)$, turning in this direction. 

\begin{lemma}\label{lem:odd-cheese}
  Let $G$ be a~graph that has an induced odd cycle $C$, and let $v \in V(G) \setminus V(C)$ have at least two neighbors in $V(C)$.
  Then there is a~subpath of $C$ that forms with $v$ an induced odd cycle of~$G$.  
\end{lemma}
\begin{proof}
  Let $v_1, \ldots, v_h$ list the neighbors of $v$ on $C$, starting at an arbitrary fixed vertex and turning in a fixed (arbitrary) direction.
  The cycles $C_1, \ldots, C_h$, each with vertex set $\{v\} \cup C[v_i \rightarrow v_{i \pmod h + 1}]$ for $i \in [h]$, respectively, have a combined number of edges of $|V(C)|+2h$; an odd number since $|V(C)|$ is odd.
  Hence at least one of $C_1, \ldots, C_h$ has odd length.
\end{proof}

Let $G$ be a graph and $v \in V(G)$.
We denote by $N_G(v)$ and $N_G[v]$, the open, respectively closed, neighborhood of~$v$ in~$G$.
For a non-negative integer $r$, we denote by $N_G^{\leqslant r}(v)$ (resp.~$N_G^{= r}(v)$) the set of vertices of $G$ at distance at~most~$r$ (resp.~exactly~$r$) from $v$.
We may call $N_G^{\leqslant r}(v)$ the ball of radius $r$ around~$v$ (in $G$).
We will rely on the observation that if a~possibly long odd cycle is contained in a~ball of small radius, then there is also a~short odd cycle (in this ball).

\begin{lemma}\label{lem:odd-cycle-in-small-ball}
  Let $G$ be a~graph, $v \in V(G)$, and $r$ be a~positive integer.
  Suppose $G[N_G^{\leqslant r}(v)]$ contains an odd cycle~$C$.
  Then $G$ (even $G[N_G^{\leqslant r}(v)]$) has an odd cycle of length at~most~$2r+1$.
  Furthermore, if $N_G^{=r}(v) \cap V(C)$ is an independent set of~$G$, then $G$ (even $G[N_G^{\leqslant r-1}(v)]$) has an odd cycle of length at~most~$2r-1$.
\end{lemma}
\begin{proof}
  Perform a~breadth-first search (BFS) starting at~$v$ in $G[N_G^{\leqslant r}(v)]$.
  By definition, this breadth-first search has exactly $r+1$ layers $L_0, L_1, \ldots, L_r$, with $L_i := N_G^{=i}(v)$ for each $i \in [0,r]$.
  At least one $L_i$ is \emph{not} an independent set, otherwise $G[N_G^{\leqslant r}(v)]$ would be bipartite, and would not contain any odd cycle.
  Say $x,y \in L_i$ are adjacent, and let $z \in L_{i'}$ be their lowest common ancestor in the BFS tree (possibly $z = v$ and $i' = 0$).
  Then, there is a~cycle of length $2(i-i')+1 \leqslant 2r+1$ in $G[N_G^{\leqslant r}(v)]$, consisting of the edge $xy$ and the paths from $z$ to $x$ and from 	$z$ to $y$ in the BFS tree. 
  Furthermore, if $N_G^{=r}(v) \cap V(C)$ is an independent set of $G$, then the odd cycle $C$ has at least one of its edges with both endpoints in some $L_i \neq L_r$, which makes an odd cycle of length at~most~$2r-1$.
\end{proof}

\subparagraph{String representations, and some useful notions and properties}

A~\emph{string} is a~subset of $\mathbb R^2$ homeomorphic to a~closed segment.
We may call \emph{substring} of a~string $s$ a~path-connected subset of $s$.
If $R$ is a~collection of strings in the plane, we denote by $G_R$ their intersection graph, with one vertex per string, and an edge between every pair of intersecting strings.
For any vertex $v \in V(G_R)$, we denote by $s_R(v)$ the string \emph{representing} $v$ in $R$. 
Conversely, we may denote by $v_R(s)$ the vertex represented by the string~$s$.
We typically drop the subscript when $R$ is clear from the context.
If $X$ is an induced subgraph of $G_R$, we may denote by $R[X]$ the subset of strings of $R$ corresponding to vertices of~$X$. 

A~\emph{face} of a~geometric arrangement made by a collection $\mathcal S$ of strings, or more generally of curves, in $\mathbb R^2$ is a~connected component of~$\mathbb R^2 \setminus \bigcup_{s \in \mathcal{S}} s$. 
A~\emph{closed face} is the topological closure of a~face.
If $D$ is a~topological disk, and more generally a~face, we denote by $\partial D$ its boundary.
For instance, the curve $\partial D$ defines two faces, one of which is finite ($D \setminus \partial D$), and admits $D$ as a~closed face.
We may denote by $\overline F$ the closure of an (open) face~$F$. 

If $P$ is an~induced path represented by strings $s_1, \ldots, s_h$ (i.e., $s_i$ and $s_j$ intersect whenever $|i-j|=1$), $a \in s_1$ and $b \in s_h$, we denote by $s[a, P, b]$ a~minimal path-connected subset of $\bigcup_{i \in [h]} s_i$ containing $a$ and $b$.
In particular $s[a, P, b]$ is a~string with endpoints $a$ and $b$, coinciding with each $s_i$ on a~substring of positive length.
Although $s[a, P, b]$ is not necessarily uniquely defined, every further claim will hold regardless of the actual choice for $s[a, P, b]$.
It can thus be thought as a~short-hand for: \emph{fix any minimal path-connected subset of $\bigcup_{i \in [h]} s_i$ containing $a$ and $b$, and denote it $s[a, P, b]$}.

\begin{figure}[h!]
\centering
\begin{tikzpicture}

  \begin{scope}[scale=0.85]
  \foreach \i/\x/\y in {1/0/0, 2/1.5/1, 3/2.4/0.9, 4/2.8/0.95, 5/3.15/0.9, 6/4.1/0.8, 7/4.6/0.75, 8/5.6/0.6, 9/6.7/0}{
    \coordinate (v\i) at (\x,\y);
  }

  \node at (0,-0.25) {$a$};
  \node at (6.7,-0.25) {$b$};
  \node at (3.3,0.2) {$P$};

  \draw[thick, purple!80!white] (v1) .. controls (0.6,0.8) .. (v2) .. controls (1.4,1.5) and (2,1) .. (2.5,1.7);

  \draw[thick, blue] (1.2,1.6) to [bend left=-15] (v2) .. controls (1.8,0.7) .. (v3) .. controls (2.7,1.4) .. (v4) .. controls (3,0.8) .. (v5) to [bend right=20] (3.9,1.4) ;

  \draw[thick, green!70!black] (2.2,1) to [bend left=-15] (v3) .. controls (2.6,0.8) .. (v4) .. controls (3,1.4) .. (v5) .. controls (3.4,0.7) .. (3.7,0.85) .. controls (3.9,1) .. (v6) .. controls (4.2,0.6) .. (v7) to [bend left=-15] (4.9,1) ;

  \draw[thick, orange] (4,0.7) to [bend left=25] (v6) .. controls (4.8,0.2) .. (v7) .. controls (4.6,1.5) .. (v8) to [bend left=-15] (5.9,0.45);

  \draw[thick] (5.1,0.9) to [bend left=-25] (5.2,0.7) to [bend left=-20] (v8) .. controls (6.3,0.5) .. (v9);

  \foreach \i/\x/\y in {1/0/0, 2/1.5/1, 3/2.4/0.9, 4/2.8/0.95, 5/3.15/0.9, 6/4.1/0.8, 7/4.6/0.75, 8/5.6/0.6, 9/6.7/0}{
    \fill (\x,\y) circle (1pt);
  }
  \end{scope}
  
  \begin{scope}[scale=0.85, xshift=8.5cm]
  \foreach \i/\x/\y in {1/0/0, 2/1.5/1, 3/2.4/0.9, 4/2.8/0.95, 5/3.15/0.9, 6/4.1/0.8, 7/4.6/0.75, 8/5.6/0.6, 9/6.7/0}{
    \coordinate (v\i) at (\x,\y);
  }

  \node at (0,-0.25) {$a$};
  \node at (6.7,-0.25) {$b$};
  \node at (3.3,0.2) {\textcolor{red}{$s[a,P,b]$}};

  \draw[thick, purple!80!white] (v1) .. controls (0.6,0.8) .. (v2) .. controls (1.4,1.5) and (2,1) .. (2.5,1.7);
  \draw[very thick, red] (v1) .. controls (0.6,0.8) .. (v2) ;

  \draw[thick, blue] (1.2,1.6) to [bend left=-15] (v2) .. controls (1.8,0.7) .. (v3) .. controls (2.7,1.4) .. (v4) .. controls (3,0.8) .. (v5) to [bend right=20] (3.9,1.4) ;
  \draw[very thick, red] (v2) .. controls (1.8,0.7) .. (v3) .. controls (2.7,1.4) .. (v4) .. controls (3,0.8) .. (v5) ;

  \draw[thick, green!70!black] (2.2,1) to [bend left=-15] (v3) .. controls (2.6,0.8) .. (v4) .. controls (3,1.4) .. (v5) .. controls (3.4,0.7) .. (3.7,0.85) .. controls (3.9,1) .. (v6) .. controls (4.2,0.6) .. (v7) to [bend left=-15] (4.9,1) ;
  \draw[very thick, red] (v5) .. controls (3.4,0.7) .. (3.7,0.85) .. controls (3.9,1) .. (v6) .. controls (4.2,0.6) .. (v7) ;

  \draw[thick, orange] (4,0.7) to [bend left=25] (v6) .. controls (4.8,0.2) .. (v7) .. controls (4.6,1.5) .. (v8) to [bend left=-15] (5.9,0.45);
  \draw[very thick, red] (v7) .. controls (4.6,1.5) .. (v8) ;
  
  \draw[thick] (5.1,0.9) to [bend left=-25] (5.2,0.7) to [bend left=-20] (v8) .. controls (6.3,0.5) .. (v9);
  \draw[very thick, red] (v8) .. controls (6.3,0.5) .. (v9);

  \foreach \i/\x/\y in {1/0/0, 2/1.5/1, 3/2.4/0.9, 4/2.8/0.95, 5/3.15/0.9, 6/4.1/0.8, 7/4.6/0.75, 8/5.6/0.6, 9/6.7/0}{
    \fill (\x,\y) circle (1pt);
  }
  \end{scope}
\end{tikzpicture}
\caption{Left: 5-vertex path as the intersection graph of strings, where the first string of the path contains point~$a$, and the last string of the path contains point~$b$.
Right: illustration of $s[a,P,b]$.}
\label{fig:s[a,P,b]}
\end{figure}
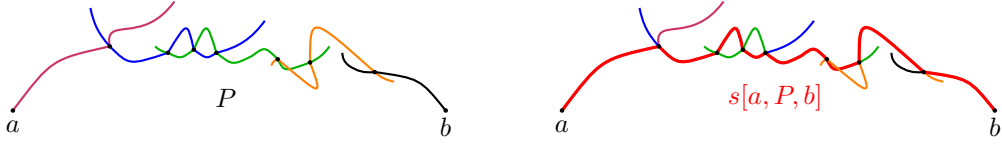

It will be convenient to extend the latter notation to induced cycles, that is, to allow $s_1$ and $s_h$ to be adjacent.
If $P$ may be a~cycle, then $s[a, P, b]$ is defined more specifically in the case when $s_1$ and $s_h$ indeed intersect.
If there is a~point $c \in s_1 \cap s_h$ and two substrings $s'_1 \subseteq s_1, s'_h \subseteq s_h$ with endpoints $a, c$ and $b, c$, respectively, such that $(s'_1 \cup s'_h) \cap \bigcup_{2 \leqslant i \leqslant h-1} s_i = \emptyset$, then $s[a, P, b] := s'_1 \cup s'_h$.
If not, then there is a~(\emph{non-self-intersecting}) string that starts at~$a$, follows non-empty substrings of~$s_1, \ldots, s_h$ in this order, and ends at~$b$.
We fix, as before, $s[a, P, b]$ to be any such string.
    
We will need the following lemma.

\begin{lemma}\label{lem:cell}
  Let $D$ be a~topological disk in the plane.
  Let $V \uplus W$ be a set of strings all contained in $D$, whose intersection graph is a cycle~$C$, such that
  \begin{compactitem} 
    \item $|W| \geqslant 2$, 
    \item every string of $W$ has one or both endpoints in $\partial D$,
    \item no other point of a~string of~$V \cup W$ intersects $\partial D$, and
    \item no two strings of $(V \cup) W$ intersect at a~point of $\partial D$.
  \end{compactitem}
  Let $s_p$ be a~string contained in~$D$ with at~least one endpoint $p$ in $\partial D$, such that $s_p$ does not intersect any string of $V \cup W$.  

  Then there are two strings $s \neq s' \in W$, and $h$ strings $s_1, \ldots, s_h \in V$ with $v(s), v(s_1), \ldots,$ $v(s_h), v(s')$ a~subpath of~$C$ (or $C$ itself), and $s_p$ is contained in a~closed face $\overline F$ of the arrangement $\{\partial D, s, s_1, \ldots, s_h, s'\}$ such that $\overline F$ does not intersect any string of $V \cup W$ outside of $s, s_1, \ldots, s_h, s'$, the string $s_{-1} \neq s_1 \in V \cup W$ intersecting $s$, and the string $s_{h+2} \neq s_h \in V \cup W$ intersecting~$s'$.  
\end{lemma}

\begin{proof}
  Let $t, t'$ be two distinct strings of $W$, with $q \in t \cap \partial D$ and $q' \in t' \cap \partial D$.
  Let $P_1, P_2$ be the two paths from $v(t)$ to $v(t')$ in~$C$.
  We define the strings $t_1 := s[q, P_1, q']$ and $t_2 := s[q, P_2, q']$.
  We denote by $\langle q,q' \rangle$ and $\langle q',q \rangle$ the two minimal path-connected subsets of $\partial D$ containing $q$ and $q'$.
  Let $F_1, F'_1$ be the two finite (open) faces with boundary $\langle q,q' \rangle \cup t_1$ and $\langle q',q \rangle \cup t_1$, respectively.
  Similarly let $F_2, F'_2$ be the two finite faces with boundary $\langle q,q' \rangle \cup t_2$ and $\langle q',q \rangle \cup t_2$, respectively.

\begin{figure}[h!]
\centering
\begin{tikzpicture}
  \draw[thick] plot [smooth cycle, tension=1] coordinates {(10,2.5)(9,3.5)(8,4.5)(6,4.9)(4,4.5)(2,4.1)(1,3.1)(0,2.5)(0.4,1.5)(2,0.5)(4,0.1)(6,0.3)(8,0.9)(9.4,1.5)};
  \node at (3,-0.1) {$\partial D$} ;

\foreach \i in {
  {(7.3,2.4) (7.5,3.5) (6.5,3)},         
  {(7.2,3.4) (5.5,3.8) (4.5,3.5)},                     
  {(5,3) (3,3.4) (2.3,3)},             
  {(4,4.1) (2.8,4) (2.5,2.5) (1,2.6) (0.8,2)}}{     
  \draw[thick] plot [smooth, tension=1] coordinates {\i};
}

\foreach \i in {
  {(1.6,1.1) (2,1.3) (4,1.1) (5,1.9)},
  {(4,2) (5,1) (6,1.5)},
  {(4.5,0.7) (6,1.7) (7,1) (7.5,2)},
  {(6.4,1.3) (8,1.5) (8.4,3) (8.6,1.6) (8,1.3)}}{ 
  \draw[thick,black!30!green] plot [smooth, tension=1] coordinates {\i};
}

\foreach \i in {
  {(9.2,2.2) (9.6,2.4) (8,3.5) (7,3)},
  {(0.4,2.4) (2.5,1.5) (1.2,1.8)}}{ 
  \draw[thick,purple] plot [smooth, tension=1] coordinates {\i};
}

\foreach \i in {
  {(1.15,0.88) (1.5,1.5) (1.8,2) (2.7,1) (3,2) (3.8,1) (4,1) (3,2.4) (2.55,1.3) (1.85,2.2) (1.2,1.65) (0.8,1.1)},
  {(8.2,1.4) (9,3) (9.5,2) (9.81,1.82)},  
  {(6,4.9) (5,3.5) (4.5,4) (4,2.5) (3,3.5)}}{ 
  \draw[thick,blue] plot [smooth, tension=1] coordinates {\i};
}

\node at (0.8,0.9) {$q$} ;
\node at (6.05,5.15) {$q'$} ;
\node at (9.8,1.58) {$q''$} ;

\node at (0.92,1.55) {\textcolor{blue}{$t$}} ;
\node at (5.9,4.4) {\textcolor{blue}{$t'$}} ;
\node at (9.3,1.8) {\textcolor{blue}{$t''$}} ;

\draw[thick, red] plot [smooth, tension=1] coordinates {(5,0.1) (5.4,0.6) (5.8,0.45)} ;
\node at (5,-0.15) {$p$} ;
\node at (5.5,0.4) {\textcolor{red}{$s_p$}} ;

\node at (5,2.5) {$C$} ;

\end{tikzpicture}
\caption{Illustration of the proof of~\cref{lem:cell}.
  The strings of $W$ are in blue.
  Say, we initially pick $t \neq t' \in W$ as depicted.
  The $v(t)$--$v(t')$ subpath of~$C$ enclosing $s_p$ contains a~string in $W$, namely~$t''$.
  As $qpq''$ turns as $qpq'$ along $\partial D$, $t'$ will be set to $t''$, and the new $P$ (green strings) is entirely in~$V$, so the process stops.
  The strings $s := t, s' := t''$, the green strings, and $\partial D$ form a~closed face containing $s_p$ that, among the strings of the rest of~$C$, only the two purple strings (may) intersect.   
}
\label{fig:cell}
\end{figure}
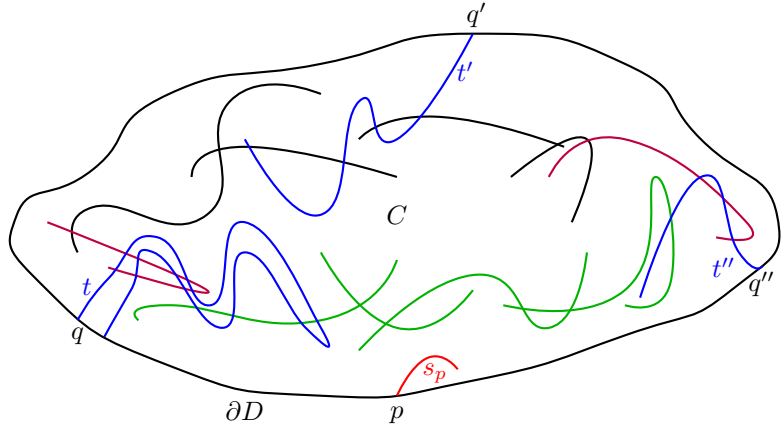

  As $C$ is an induced cycle, at~most one of $\overline{F_1}, \overline{F'_1}$ intersects the string of some vertex in $V(C) \setminus N_C[P_1]$.
  Without loss of generality, suppose that $\overline{F_1}$ does not intersect any string of $V(C) \setminus N_C[P_1]$.
  As $s_p$ does not intersect $t_1$, it is contained in the closed face $\overline{F_1}$ or in the closed face $\overline{F'_1}$.
  If $s_p$ is contained in $\overline{F_1}$, we set $P := P_1 - \{v(t),v(t')\}$ and $F := F_1$, and proceed to the next paragraph.
  Otherwise, we observe that $s_p$ is contained in $\overline{F'_2}$, and $\overline{F'_2}$ does not intersect any string of $V(C) \setminus N_C[P_2]$.
  In which case we set $P := P_2 - \{v(t),v(t')\}$ and $F := F'_2$.

  At this point, we have all the requirements of the lemma except that some strings of $P$ may be in~$W$.
  Let $P$ be $v_1, \ldots, v_\ell$ with $v_1$ a~neighbor of~$v(t)$, and $v_\ell$ a~neighbor of~$v(t')$. 
  While there is some string $t'' \in W$ with $v(t'')=v_i \in V(P)$, we let $q'' \in t'' \cap \partial D$ and distinguish two cases.
  If $q, p, q''$ turn along $\partial D$ in the same direction as $q, p, q'$, we set $t' := t''$ and $P := v_1, \ldots, v_{i-1}$.
  Otherwise we set $t := t''$ and $P := v_{i+1}, \ldots, v_\ell$.
  (Note that $|V(P)|$ decreases by at~least one in each case.) 
  When we exit the while loop, $s := t, s' := t' \in W$, and the subpath $P$ of $C$ (now only made of strings of~$V$) satisfy the lemma statement.
  A~visual recap is provided by~\cref{fig:cell}.
\end{proof}

\section{String graphs of odd girth larger than $11$ are $8$-colorable}

This section is devoted to the proof of \cref{thm:coloring}.
Let $G$ be a~string graph of odd girth strictly greater than 11, and $R$ be a~string representation of~$G$.
We may color each connected component of~$G$ independently, thus assume that $G$ is connected.
Let $u_0 \in V(G)$ be an arbitrary vertex. 
For each $i \geq 0$, let $L_i := N_G^{=i}(u_0) \subseteq V(G)$ be the set of vertices at distance exactly $i$ from $u_0$.
Note that $\{u_0\} = L_0, L_1, L_2, \ldots$ partition $V(G)$, and that there may be an edge between $x \in L_i$ and $y \in L_j$ only if $|i-j| \leqslant 1$.

Our goal is to 4-color $G[L_i]$ for each non-negative integer $i$.
We are then done, since we can use two disjoint 4-color palettes for $G[\bigcup_{i~\text{is odd}} L_i]$ and $G[\bigcup_{i~\text{is even}} L_i]$.
We fix $i$, and assume that $i \geqslant 2$ since $L_0$ is a~singleton, and $L_1$ is an independent set, due to the condition on the odd girth.
Note that it is sufficient to 4-color each connected component of~$G[L_i]$.
Let $X$ be the vertex set of any connected component of~$G[L_i]$.

\subparagraph{Definition of the topological disk $\bm{D}$.}
Let $R[X]$ be the string representation of $G[X]$ obtained by keeping the strings of $R$ corresponding to vertices of~$X$.
Let $D_0$ be the topological disk whose boundary is the boundary of the infinite face in the arrangement $R[X]$.
If $s(u_0)$ is contained in $D_0$, we draw $R$ on a~sphere, place any free point (i.e., not occupied by a~string of~$R$) of the face containing $s(u_0)$ in $R[X]$ at its north pole, and consider the stereographic projection of this string representation.
The latter is such that $s(u_0)$ is on the infinite face of the string arrangement~$R[X]$.
Thus we can in fact assume that $s(u_0)$ is not contained in $D_0$.

Note that every string of $R[X]$ is contained in $D_0$, and apart from those of $L_{i-1} \cup L_{i+1} \cup X$ no string of $R$ intersects $\partial D_0$.
Let $D \supset D_0$ be a~very slightly augmented topological disk such that $\partial D_0 \subset D \setminus \partial D$ and the property that no string outside $L_{i-1} \cup L_{i+1} \cup X$ intersects $\partial D$ is preserved.
We can also ensure that no intersection of two strings of $R$ lies on $\partial D$, and, for any point $q \in \partial D$, no string of $R$ intersects $\partial D$ at $q$ without crossing it at $q$.
We observe that every string of $L_{i-1}$ with a~neighbor in $V(X)$ intersects $\partial D$, and that every vertex of $X$ has a~neighbor in $L_{i-1}$.

\subparagraph{Definition of the auxiliary graph~$\bm{H}$.}
Let $R_H$ be the set of strings formed by $D \cap R[(L_{i-1} \cap N_G(V(X))) \cup V(X)]$, and $H$ its intersection graph.
Every vertex of $X$ corresponds to a~single vertex in $H$, whereas each vertex of $L_{i-1} \cap N_G(V(X))$ may split in several strings in $H$ as the corresponding string in $R$ may enter and exit $D$ several times.
However, the strings of $R_H$ that correspond to the~same vertex in $L_{i-1} \cap N_G(V(X))$ are pairwise non-intersecting and thus they form an independent set in~$H$.
Fix an arbitrary vertex $w \in L_{i-1} \cap N_G(V(X))$.
Let $X_j$ be the subset of vertices of $X$ at distance exactly $j$ from~$w$ in~$H$.
Again, $X_1, X_2, \ldots$ partition $X$.
By the previous remarks, it is sufficient to show that that for any positive integer $k$, the graph $G[X_k]=H[X_k]$ is bipartite. 
We fix $k$, and show this fact in the next section.

\subsection{$G[X_k]$ is bipartite}

For the sake of contradiction assume that there is an induced odd cycle $C$ in $G[X_k]$.
We denote by $v_1, \ldots, v_h$ the vertices of $C$.
For every $\ell \in [h]$, we fix some $w_\ell \in N_H(v_\ell) \setminus X$.
Such a~vertex exists since every vertex of $L_i$ has at~least one neighbor in $L_{i-1}$ (possibly split into several strings in $R_H$, at least one of which intersects $s(v_\ell)$).
Observe that $w_1, \ldots, w_h$ are not necessarily pairwise distinct.
We set $W := \{w_1, \ldots, w_h\}$, with $1 \leqslant |W| \leqslant h$.
We say that string $s(v_i)$ (and, by extension, vertex $v_i$) is \emph{simply attached} in a~cycle $C'$ containing $v_i$ if $w_i$ has only one neighbor in $V(C')$, namely~$v_i$. 
We denote by $r(w_i)$ one (arbitrary, if both endpoints are in $\partial D$) endpoint of $w_i$ that is in $\partial D$.

Our plan is to show that there is an odd cycle $C'$ contained in the ball of radius~6 around some vertex $z$ in $H$ such that $N_H^{=6}(z) \cap V(C')$ is an independent set.
This implies, by~\cref{lem:odd-cycle-in-small-ball}, the existence of an odd cycle of length at most 11 in $H$, which in turn implies, as we will see, that the same happens in~$G$.
To do this, we exhibit a~path $\widehat P$ in $H[V(C) \cup W]$ on at~most four vertices, whose strings define with $\partial D$ a~(finite) closed face $F$ containing~$s(w)$ (recall that $w$ is the arbitrary BFS root in~$H$), and an odd cycle $C'$ in $H[V(C) \cup W]$ ``mostly'' contained in $D \setminus F$.
We then show that $z$ can be chosen among the vertices of~$\widehat P$.

Let us first establish the following lemma.

\begin{lemma}\label{lem:to-simple}
  There is an induced odd cycle $C'$ in $H$ such that $V(C') \subseteq V(C) \cup W$, such that one of the following items holds
  \begin{compactitem}
  \item $C'=C$ and every string of $V(C')$ is simply attached in $C'$, or
  \item $V(C') \cap W = \{w_i\}$, for some $w_i \in W$, and $V(C') \setminus \{w_i\}$ induces a~subpath of $C$ whose internal vertices are all simply attached in~$C'$, or
  \item $|V(C') \cap W| \geqslant 2$, and there is a~subpath $x, v_i, \ldots, v_j, y$ of~$C'$ (or $C'$) such that $x \neq y \in W$, $\{v_i, \ldots, v_j\} \subset V(C)$ may be empty, $s(v_{i+1}), \ldots, s(v_{j-1})$ are all simply attached in~$C'$, and $s(w)$ is contained in a~closed face of~$\{\partial D, s(x), s(v_i), s(v_{i+1}), \ldots, s(v_{j-1}), s(v_j), s(y)\}$ that does not intersect any string of $C'$ but $s(x), s(v_i), s(v_{i+1}), \ldots, s(v_{j-1}), s(v_j), s(y)$ and the other neighbor in $C'$ of $x$ and of $y$.
  \end{compactitem}
\end{lemma}
\begin{proof}
  If all the strings of $C$ are simply attached, we are done as the first item holds.
  Otherwise there is some $w_i \in W$ with at least two neighbors in $V(C)$.
  By~\cref{lem:odd-cheese}, there is an induced odd cycle $C_1$ such that $V(C_1)$ comprises $w_i$ and a~subpath of $C$. 

Let $p \geq 1$, and suppose we have defined $C_p$.
If $C_p$ satisfies the second or the third item of the lemma, we are done; so we assume otherwise.
We will show that we can find an induced cycle $C_{p+1}$ in $H[V(C) \cup W]$ with fewer vertices in $V(C)$ than $C_p$ has.

First, suppose that $V(C_p) \cap W$ is a~singleton~$\{x\}$.
As $C_p$ does not satisfy the second outcome, there is a~vertex $v_i \in V(C_p) \setminus \{x\} \subseteq V(C)$ that is \emph{not} simply attached nor adjacent to $x$ in $C_p$.
We then obtain $C_{p+1}$ by applying~\cref{lem:odd-cheese} to the pair $C_p, w_i$.

Suppose instead that $C_p$ has at least two vertices in $W$.
By~\cref{lem:cell}, $s(w)$ is contained in a~closed face~$F$ of the arrangement $\{\partial D, s(x), s(v_i), s(v_{i+1}), \ldots,$ $s(v_{j-1}), s(v_j), s(y)\}$ with $x \neq y \in W$ such that $F$ does not intersect any string of $C'$ outside $s(x'), s(x), s(v_i), s(v_{i+1}), \ldots,$ $s(v_{j-1}),$ $s(v_j), s(y), s(y')$ where $x', y'$ are the unique vertices in $N_{C'}(x) \setminus \{v_i\}, N_{C'}(y) \setminus \{v_j\}$, respectively.
As $C_p$ does not satisfy the third item, $|j-i| \geqslant 2$ and there is some $v_{i'}$ with $i+1 \leqslant i' \leqslant j-1$ such that $s(v_{i'})$ is \emph{not} simply attached in $C_p$.
We then obtain $C_{p+1}$ by applying~\cref{lem:odd-cheese} to the pair $C_p, w_{i'}$; see~\cref{fig:odd-cheese-application}.
  
In both cases, the number of vertices of $V(C)$ present in the current odd cycle drops by at least one when going from $C_p$ to $C_{p+1}$. Thus, after at most $h = |V(C)|$ iterations, the procedure will return an odd cycle $C_q$ in $H[V(C) \cup W]$ that satisfies the second or the third statement of the lemma. 
\end{proof}  
  
  \begin{figure}[h!]
\centering
\begin{tikzpicture}
  \draw[thick] plot [smooth cycle, tension=1] coordinates {(10,2.5)(9,3.5)(8,4.5)(6,4.9)(4,4.5)(2,4.1)(1,3.1)(0,2.5)(0.4,1.5)(2,0.5)(4,0.1)(6,0.3)(8,0.9)(9.4,1.5)};
  \node at (3,-0.1) {$\partial D$} ;

\foreach \i in {                     
  {(5,3) (3,3.4) (2.3,3)},             
  {(4,4.1) (2.8,4) (2.5,2.5) (1,2.6) (0.8,2)},
  {(0.4,2.4) (2.5,1.5) (1.2,1.8)},
  {(1.6,1.1) (2,1.3) (4,1.1) (5,1.9)},
  {(4,2) (5,1) (6,1.5)},
  {(4.5,0.7) (6,1.7) (7,1) (7.5,2)}}{     
  \draw[thick,black] plot [smooth, tension=1] coordinates {\i};
}

\foreach \i in {
  {(7.3,2.4) (7.5,3.5) (6.5,3)},         
  {(7.2,3.4) (5.5,3.8) (4.5,3.5)},
  {(9.2,2.2) (9.6,2.4) (8,3.5) (7,3)},
  {(6.4,1.3) (8,1.5) (8.4,3) (8.6,1.6) (8,1.3)}}{     
  \draw[thick,black,opacity=0.4] plot [smooth, tension=1] coordinates {\i};
}

\foreach \i in {
  {(1.15,0.88) (1.5,1.5) (1.8,2) (2.7,1) (3,2) (3.8,1) (4,1) (3,2.4) (2.55,1.3) (1.85,2.2) (1.2,1.65) (0.8,1.1)},
  {(6,4.9) (5,3.5) (4.5,4) (4,2.5) (3,3.5)}}{ 
  \draw[thick,blue] plot [smooth, tension=1] coordinates {\i};
}

  \foreach \i in {
    {(8.2,1.4) (9,3) (9.5,2) (9.81,1.82)}}{ 
  \draw[thick,blue,opacity=0.4] plot [smooth, tension=1] coordinates {\i};
}

\draw[thick, red] plot [smooth, tension=1] coordinates {(5,0.1) (5.4,0.6) (5.8,0.45)} ;
\node at (6.1,0.6) {\textcolor{red}{$s(w)$}} ;

\node at (5.4,2.5) {\textcolor{gray}{$C_p$}} ;
\node at (4,2.2) {$C_{p+1}$} ;

\node at (6.4,2.1) {\textcolor{cyan}{$w_{i'}$}} ;
\node at (5.55,1.73) {$v_{i'}$} ;

\draw[thick, cyan] plot [smooth, tension=1] coordinates {(7.03,0.55) (6.3,1.1) (6.1,2.2) (6.8,3.2) (5.3,3.4) (4.9,3.6)} ;

\node at (0.9,1.5) {\textcolor{blue}{$x$}} ;
\node at (9.4,1.8) {\textcolor{blue!40!white}{$y$}} ;

\end{tikzpicture}
\caption{The string $w_{i'} \in W$ with several neighbors in $V(C_p)$, and the new odd induced cycle~$C_{p+1}$ obtained by~\cref{lem:odd-cheese}.
  \Cref{lem:cell} will then locate $s(w)$ as enclosed by $x, w_{i'}$ and some (here, three) strings of~$V$.
  For legibility, a~string may be labeled by its corresponding vertex.}
\label{fig:odd-cheese-application}
\end{figure}
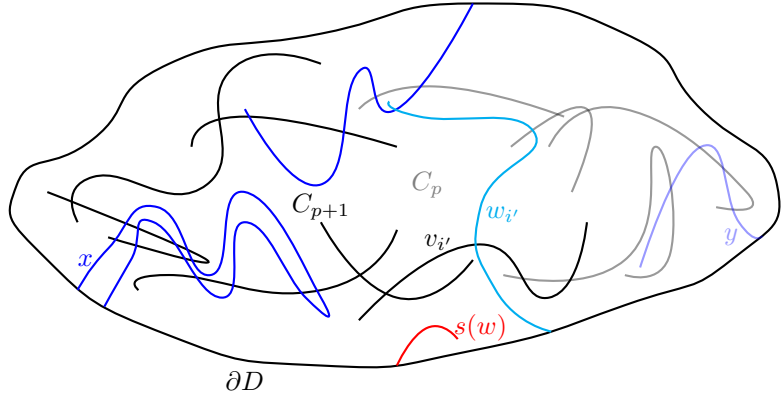

We can then obtain the announced milestone.

\begin{lemma}\label{lem:to-small-cell}
  There is an induced odd cycle $C'$ in $H[V(C) \cup W]$, two vertices $x \neq y \in W$, possibly part of $C'$, and a~subpath $P$ of $C'$ on 0, 1, or 2 vertices, all in $V(C)$, such that   
  \begin{compactitem}
  \item if $V(P) \neq \emptyset$, then $x$ is adjacent to one endpoint of $P$, and $y$, to the other (possibly equal) endpoint of~$P$, and if $V(P)=\emptyset$, then $x$ and $y$ are adjacent,
  \item one finite closed face $F_1$ of the arrangement $\{\partial D, s[r(x), xPy, r(y)]\}$ contains $s(w)$, and
  \item the other finite closed face $F_2$ contains every string of $V(C') \setminus N_{C'}[V(P) \cup \{x, y\}]$.  
  \end{compactitem}
\end{lemma}
\begin{proof}
  By~\cref{lem:to-simple}, there is an induced odd cycle $C'$ in $H[V(C) \cup W]$, two vertices $x' \neq y' \in W$, possibly part of $C'$, and a~possibly-empty subpath $P$ of $C'$ such that
  \begin{compactitem}
  \item $V(P) \subset V(C)$ and each internal vertex of~$P$ is simply attached in $C'$,
  \item $x'Py'$ is a path or a cycle (we allow $x'$ and $y'$ to be adjacent, even when $V(P) \neq \emptyset$),
  \item $\{\partial D, s[r(x'), x'Py', r(y')]\}$ has two finite closed faces $F_1 \supset s(w)$, and
  \item $F_2$ containing every string of~$V(C') \setminus N_{C'}[V(P) \cup \{x', y'\}]$.
  \end{compactitem}
  Indeed, we directly get this outcome if the third item of~\cref{lem:to-simple} holds.
  If instead the first item of~\cref{lem:to-simple} holds, we set $x' := w_i$ and $y' := w_j$ for any $v_i \neq v_j \in V(C) = V(C')$ such that $w_iw_j \notin E(H)$.
  This is possible to ensure since $H$ is triangle-free and $|V(C)| \geqslant 3$.
  The path $P$ is then chosen as the subpath of~$C$ from $v_i$ to $v_j$ such that $\partial D \cup s(x') \cup \bigcup_{u \in V(P)} s(u) \cup s(y')$ separates $s(w)$ from $V(C) \setminus N_C[P]$.
  If finally the second item of~\cref{lem:to-simple} holds, we set $x' := w_i$ for any $v_i \in V(C')$ such that $w_i \notin V(C')$, and $y'$ is defined as the only vertex of $W \cap V(C')$.
  We then set $P$ as previously.
    
  Now, while the path $P$ has at~least three vertices and $s[r(x'), x'Py', r(y')]$ passes through the strings of~$P$, let $v_i$ be an internal vertex of $P$.
  By construction, $s(v_i)$ is simply attached in $C'$.
  Let $P'$ be the subpath of $P$ going from the endpoint of $P$ neighboring $x'$ to $v_i$, and $P''$ be the subpath of $P$ going from $v_i$ to its other endpoint (neighboring $y'$).
  Then, either~$P'$ and the pair $x' \neq w_i \in W$ or $P''$ and the pair $w_i \neq y'$ satisfy the four items of the previous paragraph. 
  In the former case, we set $P := P'$ and $y' := w_i$, whereas in the latter, we set $P := P''$ and $x' := w_i$; see~\cref{fig:narrowing}.
  
    \begin{figure}[h!]
\centering
\begin{tikzpicture}
  \draw[thick] plot [smooth cycle, tension=1] coordinates {(10,2.5)(9,3.5)(8,4.5)(6,4.9)(4,4.5)(2,4.1)(1,3.1)(0,2.5)(0.4,1.5)(2,0.5)(4,0.1)(6,0.3)(8,0.9)(9.4,1.5)};
  \node at (2.1,0.2) {$\partial D$} ;

\foreach \i in {                     
  {(5,3) (3,3.4) (2.3,3)},             
  {(4,4.1) (2.8,4) (2.5,2.5) (1,2.6) (0.8,2)},
  {(0.4,2.4) (2.5,1.5) (1.2,1.8)}}{     
  \draw[thick,black] plot [smooth, tension=1] coordinates {\i};
}

\foreach \i in {                     
  {(4,2) (5,1) (6,1.5)},
  {(4.5,0.7) (6,1.7) (7,1) (7.5,2)}}{     
  \draw[thick,green!60!black] plot [smooth, tension=1] coordinates {\i};
}

\draw[thick,green!90!black] plot [smooth, tension=1] coordinates {(1.6,1.1) (2,1.3) (4,1.1) (5,1.9)};

\foreach \i in {
  {(1.15,0.88) (1.5,1.5) (1.8,2) (2.7,1) (3,2) (3.8,1) (4,1) (3,2.4) (2.55,1.3) (1.85,2.2) (1.2,1.65) (0.8,1.1)},
  {(6,4.9) (5,3.5) (4.5,4) (4,2.5) (3,3.5)},
  {(7.03,0.55) (6.3,1.1) (6.1,2.2) (6.8,3.2) (5.3,3.4) (4.9,3.6)}}{ 
  \draw[thick,blue] plot [smooth, tension=1] coordinates {\i};
}

\draw[thick,cyan] plot [smooth, tension=1] coordinates {(3.8,0.1) (4.4,1) (5.1,1.4)};

\draw[thick, red] plot [smooth, tension=1] coordinates {(5,0.1) (5.4,0.6) (5.8,0.45)} ;
\node at (6.1,0.6) {\textcolor{red}{$s(w)$}} ;

\node at (4,2.2) {$C'$} ;

\node at (0.6,0.84) {$r(x')$} ;
\node at (7.45,0.32) {$r(y')=r(y)$} ;
\node at (3.9,-0.12) {$r(x)$} ;

\node at (3.85,0.5) {\textcolor{cyan}{$x$}} ;
\node at (0.9,1.5) {\textcolor{blue}{$x'$}} ;
\node at (7.4,3) {\textcolor{blue}{$y'=y$}} ;

\end{tikzpicture}
\caption{Example of strings $x' \neq y' \in W$ and the subpath~$P$ (three green strings) of $C'$.
  The internal node of~$P$ is simply attached to~$C'$.
  This is the case when $x'$ should be updated (to the vertex of the cyan string).
  The new path $P$ (darker green) has two vertices, thus the process stops.}
\label{fig:narrowing}
\end{figure}
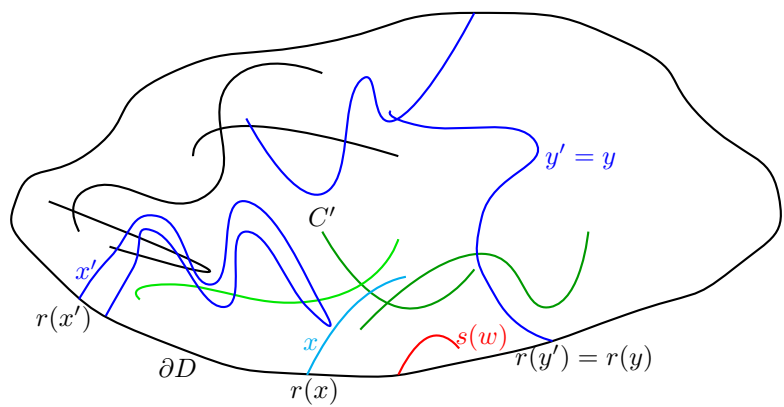

  We are done when $|V(P)| \leqslant 2$ or when $s[r(x'), x'Py', r(y')]$ is contained in $s(x') \cup s(y')$.
  In both cases, we set $x := x'$ and $y := y'$, and in the latter case, we set $P$ to be empty. 
\end{proof}

Following the notations of~\cref{lem:to-small-cell}, we define $z$ as the neighbor of~$x$ in $P$ if $V(P) \neq \emptyset$, and as $x$ otherwise.

\begin{lemma}\label{lem:close-to-C'}
  $V(C') \subseteq N_H^{\leqslant 6}(z)$.
\end{lemma}
\begin{proof}
  We keep the notations of~\cref{lem:to-small-cell}.
  First we observe that any vertex in $N_{C'}[V(P) \cup \{x, y\}]$ is at distance at~most~3 of~$z$ in~$H$.
  It remains to see that every vertex $z'$ of~$C'$ such that $s(z')$ is contained in~$F_2$ is at distance at~most~6 from~$z$ in~$H$.

  Note that every vertex of $W$ is at distance $k-1$, $k$, or $k+1$ from $w$ in~$H$, since every vertex of~$C$ is at distance exactly $k$ from $w$ in~$H$, and every vertex of~$W$ has a~neighbor in~$V(C)$.  
  As $s(w)$ is in $F_1$, a~(shortest) path from $w$ to $z'$ has to contain a~vertex $y'$ such that $s(y')$ intersects $s[r(x), xPy, r(y)]$.
  Since vertices of $V(P) \cup \{x,y\}$ are at distance at least $k-1$ from $w$, vertex $y'$ is at distance at least $k-2$ from $w$.
  Now, as vertices of $C'$ are at distance at~most $k+1$ from $w$, vertex $y'$ is at distance at most~3 from $z'$.
  In turn, $z$ is at distance at~most~3 from $y'$, and we conclude.
\end{proof}

We further show that no edge of $C'$ can lie within $N_H^{=6}(z)$.

\begin{lemma}\label{lem:6th-IS}
  $N_H^{=6}(z) \cap V(C')$ is an independent set of~$H$ or $H$ admits an odd cycle of length at~most~9.
\end{lemma}
\begin{proof}
  The proof of~\cref{lem:close-to-C'} shows that for some vertex $z' \in V(C')$ to be at distance~6 from~$z$ in $H$, there should be a~4-edge path from $y$ to $z'$.
  Thus two vertices $z', z''$ adjacent in $C'$ and at distance~6 from~$z$ in $H$ entail an odd cycle of length \emph{at~most}~9 in $H$, built from $y$, the corresponding two (non-necessarily edge-disjoint) 4-edge paths, $z'$, and $z''$. 
\end{proof}

\cref{lem:close-to-C',lem:6th-IS,lem:odd-cycle-in-small-ball} imply the existence of an odd cycle $C_o$ of length at~most 11 in $H$.
We finally see that this yields a~contradiction, as we can build an odd cycle of at~most the same length in~$G$.

\begin{lemma}\label{lem:odd-cycle-in-G}
Let $C_o$ be an odd cycle in $H$.
Then $G$ has an odd cycle of length at~most~$|V(C_o)|$.
\end{lemma}
\begin{proof}
  We initialize a~represented odd cycle $\widehat C$ to $C_o$ realized by the corresponding strings of~$R_H$. 
  We make an induction on the number of strings of~$\widehat C$ not part of~the representation $R$ of~$G$.
  When this number is 0, we conclude since $G$ contains the cycle $\widehat C$ as a~subgraph.
  Let $y_1, \ldots, y_q \in V(\widehat C)$ be all the vertices of $\widehat C$ whose strings are substrings of $s(y)$ for some $y \in L_{i-1} \cap N_G(V(X))$.

  We further assume that starting at $y_1 \in V(\widehat C)$, and turning in some fixed (arbitrary) direction along $\widehat C$, one encounters $y_1, y_2, \ldots, y_q$ in this order.
  Let, for every $a \in [q]$, $\ell_a$ be the number of edges of $\widehat C[y_a \rightarrow y_{a \pmod q + 1}]$.
  As, $|V(\widehat C)| = \sum_{a \in [q]} \ell_a$, at least one $\ell_a$ is odd.
  Recall that $\{y_1, \ldots, y_q\}$ is an independent set in $H$, hence every $\ell_a$ is strictly greater than~1.
  Then the string $s(y)$ and those of the internal vertices in the path between $y_a$ and $y_{a \pmod q + 1}$ form an odd cycle of length at~least~3 and at~most~$|V(\widehat C)|$.
  This defines the new represented odd cycle $\widehat C$, and concludes the proof.
\end{proof}

\subsection{8-coloring algorithm}

To claim \cref{thm:coloring}, we finally need to check that, given a~representation $P, \mathcal S$ of a~string graph $G$ of odd girth larger than 11, one can compute an 8-coloring of~$G$ in time polynomial in $|V(P)|$. Recall that $P$ is a~planar graph and $\mathcal S$ is a~set of non-empty connected sets in $P$ in one-to-one correspondence with $V(G)$ such that two distinct connected sets of $\mathcal S$ intersect if and only if the corresponding vertices of $G$ are adjacent.
Note that, as $G$ is in particular triangle-free, $|V(G)| \leqslant 2|V(P)|$.

We remind the reader that we compute one BFS in~$G$, and less than~$|V(G)|$ BFSes in auxiliary graphs $H$ of size at~most~$(|V(P)|+1)|V(G)|$.
After this we simply 2-color bipartite graphs whose combined number of vertices is at~most $|V(G)|$.
Thus, we shall just detail how to compute each auxiliary graph~$H$.

Let $X$ be the connected component of $G[L_i]$ giving rise to~$H$.
We start by adding every vertex of~$X$ to $V(H)$.
We compute the set $Y \subseteq V(P)$ of all the vertices contained in an element of $\mathcal S$ corresponding to a~vertex of~$X$.
Let $u' \in V(P)$ be a~vertex in the connected set of~$u_0$.
Let $Y' \subseteq Y$ be the vertices $v \in Y$ for which there is a~path in~$P$ from $u'$ to $v$ whose internal vertices are all in $V(P) \setminus Y$.
Note that $Y'$ can be computed in polynomial time in $|V(P)|$ by checking if $u'$ and $y$ are in the same connected component of $P-(Y \setminus \{v\})$.
The set $Y'$ is the combinatorial counterpart of $\partial D$.
In particular, we do not need to change the representation if $u'$ is in a~finite face ``made by~$Y$.''
This was merely helpful in the proof correctness.

For each vertex $v \in L_{i-1} \cap N_G(V(X))$ (alternatively we can only keep at~least~one neighbor per vertex of~$X$), we add to $V(H)$ one vertex for each connected component of $S \cap Y$ in $P$, where $S$ is the connected set of~$v$.
Note that $S' \cap Y' \neq \emptyset$ for every vertex set $S'$ of such a~connected component.
This step adds to $H$ fewer than $|V(G)|\cdot|V(P)|$ vertices.
Graph $H$ is finally defined as the intersection graph of all the connected sets of vertices in $V(H)$, which can be computed in time polynomial in~$|V(P)|$.

\section{Approximation algorithm for \textsc{Vertex Cover}}

For a graph $G$, let $\vcp(G)$ denote the size of a minimum vertex cover in $G$.
We will use the following result of Chleb\'ik and Chleb\'ikov\'a~\cite{DBLP:conf/swat/ChlebikC04}, which is a slight strengthening of the well-known Nemhauser--Trotter theorem~\cite{DBLP:journals/mp/NemhauserT74}.

\begin{theorem}[Chleb\'ik and Chleb\'ikov\'a~\cite{DBLP:conf/swat/ChlebikC04}, Nemhauser and Trotter~\cite{DBLP:journals/mp/NemhauserT74}]\label{thm:nemhausertrotter}
Given a graph $G$, one can compute in polynomial time a partition of $V(G)$ into $V_0,V_{1/2}, V_1$, such that
\begin{compactenum}
\item there are no edges between $V_0$ and $V_{1/2}$ or within $V_0$, 
\item $\vcp(G[V_{1/2}]) \geq \frac{1}{2} |V_{1/2}|$, and 
\item every minimum vertex cover $S$ of $G$ satisfies $V_1 \subseteq S \subseteq V_1 \cup V_{1/2}$.
\end{compactenum}
\end{theorem}

Finally, we are ready to prove \cref{thm:vcstrings}, which we restate for convenience.

\thmvcstrings*

\begin{proof}
  Let $G$ be the input graph, given along with a representation.
The algorithm has three phases.

\subparagraph{Phase one.} We initialize $X = \emptyset$.
If $G$ contains an odd cycle of length at most 11, we include all its vertices into $X$ and remove them from the graph.
We repeat this step exhaustively; clearly this can be done in polynomial time. 
Let $G'$ be the graph obtained after the last iteration of the process. 

\subparagraph{Phase two.} We call the algorithm from \cref{thm:nemhausertrotter} on the graph $G'$ in order to obtain three sets $V_0,V_{1/2}$, and $V_1$.
We denote $Y = V_1$ and $G'' = G'[V_{1/2}]$.
Note that $G''$ is a string graph with odd girth larger than 11 and the representation of $G''$ can be easily obtained from the representation of $G$ by removing strings (or connected sets) corresponding to vertices in $V(G) \setminus V(G'')$.

\subparagraph{Phase three.} 
We apply the algorithm from \cref{thm:coloring} to find a proper coloring of $G''$ with at most 8 colors.
Let $c$ be the color that appears most frequently, and let $Z$ be set of vertices of $G''$ \emph{not} colored $c$.
Clearly $|Z| \leq \frac{7}{8} |V(G'')|$, so $|V(G'')| \geq  \frac{8}{7} |Z|$.
The algorithm returns $Q = X \cup Y \cup Z$.

\subparagraph{Analysis.}
First, let us argue that $Q$ is indeed a vertex cover.
Since $X \subseteq Q$, it is enough to show that $Q \cap V(G')=Y \cup Z$ is a vertex cover of $G'$.
Note that by the first property in \cref{thm:nemhausertrotter} and since $V_1 = Y$, all edges of $G'$ not contained in $G''$ are covered by $Q$.
So we are left with showing that $Q \cap V(G'') = Z$ is a vertex cover of $G''$.
This is clearly true, as the complement of $Z$ in $G''$ is an independent set (of color $c$).

Now let us analyze the approximation factor.
Let $S$ be an optimum solution, i.e., a vertex cover of $G$ of size $\vcp(G)$.
Note that for each odd cycle $C$ removed in the first phase, $S \cap C$ must contain at least $\frac{|C|}2$ vertices in order to cover all the edges of $C$.
As each removed cycle has at most 11 vertices, we conclude that $|S \cap X| \geq \frac{6}{11} |X|$.

Note that $S' \setminus X$ is a vertex cover of $G'$.
By the third  property in \cref{thm:nemhausertrotter} we observe that $Y = V_1 \subseteq S$, and so $S \setminus (X \cup Y)$ is a vertex cover of $G''$.
By the second  property in \cref{thm:nemhausertrotter} we obtain that $|S \setminus (X \cup Y)| \geq \frac{1}{2} |V(G'')|$.

Summing up, we obtain
\begin{align*}
|S| = & \; |S \cap X| + |S \cap Y| + |S \setminus (X \cup Y)| \geq \frac{6}{11} |X| + |Y| +  \frac{1}{2} |V(G'')|  \\ 
\geq & \; \frac{6}{11} |X| + |Y| + \frac{4}{7} |Z| \geq \frac{6}{11} |X \cup Y \cup Z| = \frac{6}{11} |Q|, 
\end{align*}
which means that $|Q| \leq \frac{11}{6} |S| = \frac{11}{6} \vcp(G)$. This completes the proof.
\end{proof}

The proof above is easily adapted to show \cref{thm:promise}.
The first two phases remain unchanged.
Let $\mathcal{C}$ be the family of odd cycles found in phase one and let $Y$ be the set found in phase two.
Once we reach phase three, we do not call the coloring algorithm on $G''$, as its running time might not be polynomial in~$|V(G)|$.
Instead, we check if $\sum_{C \in \mathcal{C}} \lceil |V(C)|/2 \rceil + |Y| + |V(G'')|/2 >k$ and, if so, we reject the instance.
Note that the sum in the expression above is a lower bound on $\vcp(G)$, so this step is correct.
Otherwise, we accept the instance. The approximation guarantee is estimated as in the proof of \cref{thm:vcstrings}.

\end{document}